\documentclass[final,5p,times,twocolumn]{elsarticle}
\usepackage[T1]{fontenc}
\usepackage[utf8]{inputenc}
\usepackage[english]{babel}
\usepackage{amsmath,amssymb,amsthm,mathtools}
\usepackage{tikz}
\usetikzlibrary{arrows.meta,positioning,calc}
\usepackage{microtype}
\usepackage{enumitem}
\usepackage{natbib}
\usepackage{hyperref}
\hypersetup{hidelinks}
\biboptions{authoryear,round}
\newtheorem{definition}{Definition}
\newtheorem{proposition}{Proposition}
\journal{arXiv.org}
\begin{document}
\begin{frontmatter}
\title{Reflexivity from Hierarchical Causality}
\author[uct]{Tim Gebbie}
\ead{tim.gebbie@uct.ac.za}
\address[uct]{Department of Statistical Sciences, University of Cape Town, Rondebosch 7701, Western Cape, South Africa}
\begin{abstract}
Complex systems are often organised into hierarchies whose internal interactions are stronger or faster than interactions across levels. When markets are treated as genuinely multilevel systems it becomes natural to represent them as systems with hierarchical causality. Here we show that reflexivity can be formulated within such a discrete hierarchical causal system; but one in which a higher-level actor state restricts the lower-level transition kernels that remain admissible. Then event dynamics can be separated from calendar embeddings: a set-valued actor-conditioned correspondence can be used to define the admissible family of event kernels, while joint state and waiting-time laws can be used to determine compatible timing to then natural demonstrate reflexivity. A selected event-state law need not determine a unique calendar embedding. Locally, uniqueness of the joint event or timing specification requires uniqueness of both the admissible event kernel and its compatible timing law. Reflexivity is thus the endogenous closure of a hierarchical constraint loop, while timing and projection can generate calendar-time memory or causal ambiguity even for Markov event dynamics.
\end{abstract}
\end{frontmatter}
\section{Hierarchical causality}
\label{sec:hc}

Complex systems are often organised into modules whose internal interactions are stronger or faster than interactions across modules. Following \citet{Simon1962}, nearly decomposable systems retain a degree of dynamical autonomy while interacting more weakly across levels or modules. \citet{May1972} gives a complementary stability result: for sufficiently random coupling, increasing system size, connectance, or interaction strength can drive a large system beyond a stability threshold. This does not imply that hierarchy must follow, but together the two results motivate organisation with strong internal cohesion and weaker external coupling rather than unrestricted high-dimensional coupling.

Even weak coupling can be causally important when a sparse actor interface changes the set of admissible lower-level transitions. Following \citet{AulettaEllisJaeger2008}, \citet{Ellis2012} and \citet{Gebbie2026}, the higher level acts by constraining lower-level possibilities rather than as a second physical force. \citet{Flack2017} gives an operational precedent in which components tune their behaviour in response to estimates of collectively produced macroscopic properties, while \citet{MontevilMossio2015} characterise biological organisation as closure among mutually dependent constraints. 

Hierarchical causality was introduced in a financial setting by \citet{WilcoxGebbie2014}, where markets are treated as genuinely multilevel systems ranging from local agents and traded assets through markets and common factors to institutional, regulatory, and social structures. In this setting agents implement local dynamics at particular levels of organisation, while actors instantiate causal roles that constrain those dynamics across levels. A discrete formulation of this structure is given by \citet{Gebbie2026}. Write the hierarchical causal system as
\begin{equation}
\Psi=(H,D,C,U),
\label{eq:hc-system}
\end{equation}
where $H$ carries hierarchical aggregation, $D$ carries local dynamics, $C$ carries actor-instance constraints, and $U$ carries level-specific event coordination. Aggregation produces higher-level states from lower-level activity. Top-down causation enters when higher-level organisation changes the lower-level dynamics that remain admissible. 

In Hierarchical Causality, local agents carry within-level dynamics, while a top-down causal role is represented by an actor $a\in\mathcal A_r$ through an actor instance
\[
A=(a,r,c,\ell,B_\ell,\eta).
\]
The actor instance records the actor, its causation class $c=\gamma(a)$, the lower-level target subsystem $B_\ell$, and the interface $\eta$ through which the constraint is implemented. The distinction between agent and actor is therefore one of causal role rather than physical identity. We fix one actor instance $A$ below; its level, target subsystem and interface are then fixed unless written explicitly. Here the higher-level action is an actor-conditioned admissible-kernel correspondence; event coordination and calendar embedding remain separate.

Figure~\ref{fig:hc-architecture} visualises the schematic representation using a hierarchical causal diagram. The horizontal and vertical routes distinguish lower-level transition from aggregation, while the dashed map represents the actor-instance constraint on the admissible target-subsystem kernel; it does not add a second equation of motion.

\begin{figure}[htbp]
\centering
\begin{tikzpicture}[
  x=1cm,y=1cm,>=Latex,
  every node/.style={font=\footnotesize},
  block/.style={align=center,minimum width=2.55cm,minimum height=0.92cm,inner xsep=2pt},
  lab/.style={font=\footnotesize,fill=white,inner sep=1.1pt,align=center}
]
    \node[block] (xr)     at (0,3.0) {$X_{r,n_r}$\\ higher-level state};
    \node[block] (xrdesc) at (5.5,3.0) {$\Pi_{\ell}(X_{\ell,n_{\ell}+1})$\\ aggregate description};
    \node[block] (xl)     at (0,0) {$X_{\ell,n_{\ell}}$\\ lower-level state};
    \node[block] (xlp)    at (5.5,0) {$X_{\ell,n_{\ell}+1}$\\ lower-level state};
    \draw[->] (xl.east) -- node[lab,below=3pt] {$K^{B_\ell}_{\ell,n_\ell}$} (xlp.west);
    \draw[->] (xl.north) -- node[lab,left=3pt] {$\Pi_{\ell}$} (xr.south);
    \draw[->] (xlp.north) -- node[lab,right=3pt] {$\Pi_{\ell}$} (xrdesc.south);
    \draw[->] (xr.east) -- node[lab,above=3pt] {descriptive aggregate route} (xrdesc.west);
    \draw[->,dashed,bend left=8]
      (xr.south east)
      to node[lab,midway,text width=3.4cm,xshift=-0.20cm]
      {$D^{A}_{r\rightarrow\ell}$ constraint on\\ $K^{B_\ell}_{\ell,n_\ell}$}
      (xlp.north west);
    \node[align=center,font=\footnotesize,text width=7.8cm] at (2.75,-1.02)
      {Non-commutation marks the gap between aggregation equivalence and causal equivalence.};
\end{tikzpicture}
\caption{Hierarchical causal structure.}
\label{fig:hc-architecture}
\end{figure}
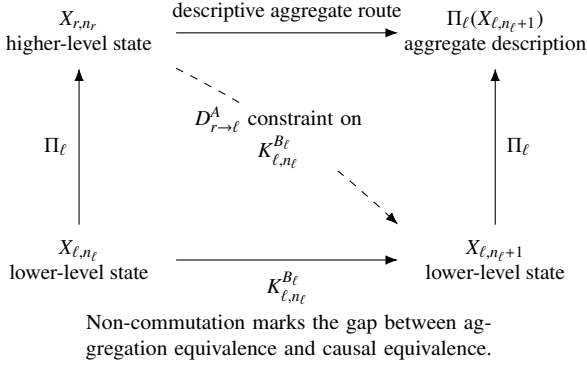

\section{Causal structure}
\label{sec:causal-structure}

The structural causal models of \citet{Pearl2009} provide semantics for intervention and counterfactual dependence. For a structural variable $X$, the intervention $\operatorname{do}(X=x)$ replaces its structural assignment by $x$ while leaving the other assignments unchanged, and is therefore defined relative to a specified structural model. Hierarchical causality addresses a prior organisational question: where the causal role is instantiated, the interface through which it acts, and which lower-level transition structures remain admissible after the higher-level constraint is applied.

This is important because aggregation equivalence is not causal equivalence \citep{Gebbie2026}. Two lower-level states may have the same aggregate description and yet differ in their admissible next transitions once an actor-level constraint is active. Selecting one admissible kernel fixes the transition law. A structural realisation of that law is then needed to apply intervention semantics in the sense of \citet{Pearl2009}. The hierarchical description keeps explicit the state, actor role, interface, and admissible lower-level dynamics from which that transition law is obtained.

\citet{RubensteinEtAl2017} formalise exact transformations between structural equation models, including micro-to-macro aggregation and dynamical-to-stationary descriptions, while \citet{BeckersHalpern2019} develop causal abstractions tied to mappings of interventions. Both begin with specified causal models at the levels being related. Here aggregation together with actor state can instead leave a family of admissible lower-level transition laws before a unique structural realisation has been fixed.

Causal emergence, following \citet{HoelAlbantakisTononi2013}, asks whether a macro-description can have greater causal effectiveness than the micro-description from which it is constructed. Effective information measures this by perturbing the specified causal model across its possible states; a coarse-grained macro model can gain effective information when increased determinism or reduced degeneracy outweighs the smaller state space. \citet{Hoel2026} extends this question across a hierarchy of coarse-grained scales. Their question concerns where causal interactions are effective; here the question is how a higher-level state changes which lower-level transitions remain admissible.

\section{Reflexivity}
\label{sec:reflexivity}

Let $n=0,1,2,\ldots$ index the reflexive event sequence used below, let $X_n$ denote the corresponding lower-level state, and let $\Pi$ denote the aggregation map from the lower-level state to the higher-level state. In Hierarchical Causality a common event index $m$ is coordinated across levels by $n_\ell=U_\ell(m)$, so different levels need not update synchronously. The single index $n$ used here labels the event sequence on which the reflexive loop is written. Define the aggregate state by
\begin{equation}
M_n=\Pi(X_n).
\label{eq:aggregate-state}
\end{equation}
We allow the higher-level constraint to be set-valued. In Hierarchical Causality the fixed actor instance $A$ induces $D^A_{r\to\ell}$ on admissible lower-level kernels. We fix the target level, subsystem and interface, and assume that the relevant higher-level dependence acts through an actor state in $S_A$. Let $\mathcal D$, distinct from the dynamics block $D$ in Eq.~\eqref{eq:hc-system}, map actor state to the transition kernels that remain admissible. First consider the special case $\mathcal I=\mathrm{id}$, so that $s_n=M_n$. Then
\begin{equation}
K_n\in \mathcal D\!\left(M_n\right),
\label{eq:admissible-kernel-correspondence}
\end{equation}
where $K_n$ is the realised lower-level transition kernel at event $n$. The state $X_{n+1}$ is the lower-level state after that transition. Reflexivity arises when the admissible-kernel correspondence is itself conditioned by a state generated from the lower-level dynamics,
\begin{equation}
X_n\xrightarrow{\Pi}M_n\longrightarrow \mathcal D(M_n)
\ni K_n\longrightarrow X_{n+1}.
\label{eq:reflexive-constraint-map}
\end{equation}
Equation~\eqref{eq:reflexive-constraint-map} is already a feedback relation, but the feedback acts through the admissible dynamics rather than through an additional force.

Ordinary feedback may remain within a single dynamical level. Reflexivity becomes hierarchical when a state generated from below conditions an actor-level constraint that subsequently changes the lower-level dynamics that remain admissible. Let $\mathcal I$ denote the interpretation map from the aggregate state to the actor-state space $S_A$, and define the actor-state sequence $(s_n)_{n\ge0}$ by
\begin{equation}
s_n=\mathcal I(M_n),\qquad s_n\in S_A.
\label{eq:actor-state}
\end{equation}
The actor identity $a$, causation class $c$, target subsystem $B_\ell$ and interface $\eta$ are fixed, while $s_n$ varies with the aggregate state. The interpretation map $\mathcal I$ is distinct from the interface $\eta$: $\mathcal I$ maps the aggregate description into actor state, while $\eta$ implements the top-down restriction on lower-level dynamics. The causal map is then
\begin{equation}
X_n\xrightarrow{\Pi}M_n\xrightarrow{\mathcal I}s_n
\longrightarrow \mathcal D(s_n)\ni K_n\longrightarrow X_{n+1}.
\label{eq:actor-reflexive-map}
\end{equation}
The rule governing subsequent dynamics is endogenous because it depends on a state generated by those dynamics. This gives a hierarchical-causal interpretation of Sorosian reflexivity \citep{PolakowGebbieFlint2026,Soros2013}. Earlier mathematical models already formalise financial self-reference: \citet{WyartBouchaud2007} study strategies that alter the correlations from which those strategies are formed, while \citet{Palatella2010} gives an explicit reflexive price model motivated by the same idea. 

Here the feedback enters through an actor-conditioned admissible transition law, while its calendar-time representation remains separate. \citet{PolakowGebbieFlint2026} argue that causal chains in open, self-referential markets are malleable by market actors and may coexist in ways that are indistinguishable without intervention. Here we interpret that problem through hierarchical causality: lower-level activity produces the market state, while actor interpretation of that state changes the subsequent admissible dynamics. A non-singleton $\mathcal D(s_n)$ represents competing admissible causal chains schematically. Their epistemic argument motivates the problem; the correspondence $\mathcal D$ is introduced here.

A non-singleton $\mathcal D(s_n)$ leaves the lower-level transition law set-valued until a selection rule is supplied, before the event history is mapped into calendar time.

\section{Non-unique time}
\label{sec:nonunique-time}

Equation~\eqref{eq:actor-reflexive-map} is indexed by events, whereas valuation, risk measurement and many empirical causal claims are expressed in calendar time. Hierarchical Causality \citep{Gebbie2026} uses $U_\ell(m)$ to coordinate level-specific event counts; this is distinct from the event-to-calendar time change. In the operational-time notation of \citet{AngstmannGebbieAdjoint2026}, $u=U(t)$ maps calendar time $t$ to operational time $u$, with activity rate $\alpha(t)=\mathrm dU(t)/\mathrm dt$. More generally, \citet{AngstmannGebbie2026} argue that event, operational and calendar time need not determine a unique market representation. Temporal resolution can also distort an inferred causal structure when a unique process is assumed \citep{HyttinenEtAl2016}; the issue here is prior to inference because the event-state dynamics may themselves admit more than one compatible calendar embedding.

Let $\tau_{n+1}>0$ be the waiting time to the next event. For a generic actor state $s\in S_A$, current lower-level state $x$ and selected admissible event-time kernel $K$, let $\mathcal J_{s,K}(x,\cdot)$ denote a realised joint kernel for the next state and waiting time. For a measurable next-state set $B$ and Borel set $I\subset(0,\infty)$,
\begin{equation}
\Pr\!\left(X_{n+1}\in B,\,\tau_{n+1}\in I\mid X_n=x,s_n=s,K_n=K\right)=\mathcal J_{s,K}(x,B\times I).
\label{eq:joint-event-kernel}
\end{equation}
The state marginal determines the event-time transition law and the waiting-time marginal determines when the transition becomes visible in calendar time; no factorisation of these two marginals is assumed. Consistency with the selected event-time kernel requires
\begin{equation}
\mathcal J_{s,K}\!\left(x,B\times(0,\infty)\right)
=K(x,B).
\label{eq:joint-kernel-state-marginal}
\end{equation}
Thus the joint kernel extends rather than replaces the selected event-time law. For calendar time $t\ge0$, set $T_0=0$ and define the event epochs $T_n$, the counting process $N(t)$, and the observed state $Y_t$ by
\begin{equation}
T_n=\sum_{j=1}^{n}\tau_j,\qquad N(t)=\max\{n:T_n\le t\},\qquad Y_t=X_{N(t)}.
\label{eq:event-calendar-map}
\end{equation}
We assume non-explosion, $T_n\to\infty$ almost surely. Equation~\eqref{eq:event-calendar-map} gives the event-to-calendar embedding of the discrete event chain. Direct Bochner subordination is a special continuous-operational-time construction that composes a time-homogeneous operational Markov process with an independent increasing subordinator. Inverse-subordinator and renewal constructions can generate calendar-time memory \citep{MeerschaertStraka2013,AngstmannGebbie2026}. Appendix~\ref{app:event-calendar} separates these cases explicitly.

Equation~\eqref{eq:joint-event-kernel} describes one realised joint event law after the event-time kernel has been selected. Non-unique time means that the event-state dynamics need not fix a unique compatible joint state/waiting-time law before a timing specification is selected. Let $\Theta$ index the compatible timing specifications and write $\mathcal J^{\vartheta}_{s_n,K}$ for the corresponding family of joint laws sharing the selected event-state marginal $K$. For a timing rule $\vartheta$ applied consistently along the event sequence, the resulting sequence of compatible joint laws induces holding times and hence a clock $T^{\vartheta}=(T_n^{\vartheta})_{n\ge0}$, with $N_{\vartheta}(t)=\max\{n:T_n^{\vartheta}\le t\}$. At the induced clock level write
\begin{equation}
\mathfrak T=\{T^{\vartheta}:\vartheta\in\Theta\},
\qquad
Y_t^{\vartheta}=X^{\vartheta}_{N_{\vartheta}(t)}.
\label{eq:clock-family}
\end{equation}
Here $\mathfrak T$ is the induced admissible clock family and $Y_t^{\vartheta}$ is the corresponding calendar-time process. A usable calendar representation may require selection of one representative, or identification of representations that are equivalent for the causal, observational, or financial question at hand. Appendix~\ref{app:compatible-timing} gives the joint-law formulation explicitly.\footnote{A related mathematical construction appears in \citet{Wilcox2014LinearRelations}, where a multivalued linear operator is represented as a linear relation and the multivalued part is removed by a natural quotient to obtain a single-valued operator. Although the market-time problem considered here is not necessarily a linear-relation problem, this is a useful and powerful analogy because a single-valued calendar description may require an explicit selection or quotient of an underlying multivalued relation.}

An admissible family of event kernels is not yet a realised event law, and a fixed event kernel need not determine a unique calendar embedding. Local joint event/timing uniqueness therefore requires uniqueness at both stages: selection of the admissible event kernel and specification of the compatible joint state/waiting-time law; Proposition~\ref{prop:local-event-timing-uniqueness} makes this precise.

At event epochs the calendar-time reflexive map can be written without introducing an arbitrary calendar increment,
\begin{equation}
\begin{aligned}
Y_{T_n}=X_n
&\xrightarrow{\Pi}M_n
\xrightarrow{\mathcal I}s_n
\longrightarrow \mathcal D(s_n)\ni K_n
\\
&\longrightarrow \mathcal J_{s_n,K_n}
\longrightarrow (X_{n+1},\tau_{n+1})
\longrightarrow T_{n+1}
\longrightarrow Y_{T_{n+1}}.
\end{aligned}
\label{eq:calendar-reflexive-map}
\end{equation}
Equation~\eqref{eq:calendar-reflexive-map} composes the closed event-level feedback with its calendar embedding. The higher-level state produced by the event stream conditions the actor state that changes the admissible kernel and, potentially, the waiting-time law from which the next calendar state is formed, even though the primitive transition mechanism remains local in event index.

Figure~\ref{fig:reflexive-event-dynamics} visualises the event-level reflexive mechanism.

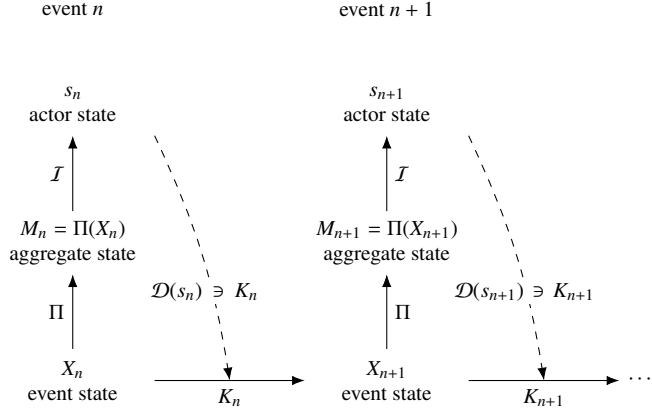
\begin{figure}[htbp]
\centering
\begin{tikzpicture}[
  x=1cm,y=1cm,>=Latex,
  every node/.style={font=\footnotesize},
  block/.style={align=center,minimum width=2.15cm,minimum height=0.82cm,inner xsep=2pt},
  lab/.style={font=\footnotesize,fill=white,inner sep=1.0pt,align=center}
]
    \node[font=\footnotesize] at (0,4.90) {event $n$};
    \node[font=\footnotesize] at (4.15,4.90) {event $n+1$};
    \node[block] (sn)  at (0,3.65) {$s_n$\\ actor state};
    \node[block] (snp) at (4.15,3.65) {$s_{n+1}$\\ actor state};
    \node[block] (mn)  at (0,1.82) {$M_n=\Pi(X_n)$\\ aggregate state};
    \node[block] (mnp) at (4.15,1.82) {$M_{n+1}=\Pi(X_{n+1})$\\ aggregate state};
    \node[block] (xn)  at (0,0) {$X_n$\\ event state};
    \node[block] (xnp) at (4.15,0) {$X_{n+1}$\\ event state};
    \coordinate (contend) at (7.225,0);
    \node[font=\footnotesize,inner sep=0pt] (cont) at (7.52,0) {$\cdots$};
    \draw[->] (xn.north) -- node[lab,left=2pt] {$\Pi$} (mn.south);
    \draw[->] (mn.north) -- node[lab,left=2pt] {$\mathcal I$} (sn.south);
    \draw[->] (xnp.north) -- node[lab,right=2pt] {$\Pi$} (mnp.south);
    \draw[->] (mnp.north) -- node[lab,right=2pt] {$\mathcal I$} (snp.south);
    \coordinate (kn) at ($(xn.east)!0.5!(xnp.west)$);
    \coordinate (knp) at ($(xnp.east)!0.5!(contend)$);
    \draw[->] (xn.east) -- node[lab,below=2pt] {$K_n$} (xnp.west);
    \draw[->] (xnp.east) -- node[lab,below=2pt] {$K_{n+1}$} (contend);
    \draw[->,dashed,bend left=8]
      (sn.south east)
      to node[lab,pos=0.66,text width=2.25cm,xshift=-0.10cm]
      {$\mathcal D(s_n)\ni K_n$}
      (kn);
    \draw[->,dashed,bend left=8]
      (snp.south east)
      to node[lab,pos=0.66,text width=2.35cm,xshift=-0.03cm]
      {$\mathcal D(s_{n+1})\ni K_{n+1}$}
      (knp);
\end{tikzpicture}
\caption{Reflexive event dynamics. The event state $X_n$ aggregates to $M_n$, which is interpreted as the actor state $s_n$. Reflexivity is explicit in $\mathcal D(s_n)\ni K_n$: the actor state generated from the current event history restricts the kernel governing the next event. The same restriction recurs at event $n+1$, where $s_{n+1}$ restricts $K_{n+1}$, making the forward-time reflexive iteration explicit.}
\label{fig:reflexive-event-dynamics}
\end{figure}

Figure~\ref{fig:reflexive-calendar-embedding} extends the same reflexive mechanism to compatible calendar embeddings.

\begin{figure}[t]
\centering
\begin{tikzpicture}[
  x=1cm,y=1cm,>=Latex,
  every node/.style={font=\footnotesize},
  lab/.style={font=\footnotesize,fill=white,inner sep=1.0pt,align=center},
  desc/.style={font=\footnotesize,align=center},
  panel/.style={font=\footnotesize\bfseries,anchor=west}
]
    \node[panel] at (0.0,7.75) {(a) Reflexive timing specification};

    \node (mn) at (0.55,6.52) {$M_n$};
    \node[desc,anchor=west] (mndesc) at ([yshift=-0.54cm]mn.west) {aggregate state};

    \node[desc] (actordesc) at (2.75,7.02) {actor state};
    \node (actor) at (2.75,6.52) {$s_n=\mathcal I(M_n)$};

    \node (restricted) at (5.15,6.52) {$\mathcal D(s_n)\ni K_n$};
    \node[desc] (restricteddesc) at (5.15,5.98) {admissible event kernel};

    \node (timingfamily) at (7.55,6.52) {$\mathfrak J(s_n,K_n)$};
    \node[desc,anchor=east] (familydesc) at ([yshift=0.50cm]timingfamily.east) {compatible timing family};

    \draw[->] (mn) -- (actor);
    \draw[->,dashed] (actor) -- (restricted);
    \draw[->] (restricted) -- (timingfamily);

    \node[align=center,text width=8.10cm] at (4.05,5.25)
      {$\vartheta_j\in\Theta,\quad
        \mathcal J^{\vartheta_j}_{s_n,K_n}\in\mathfrak J(s_n,K_n),\quad
        j=1,2,$ choices used in (b).};

    \node[panel] at (0.0,4.25) {(b) Compatible calendar embeddings};
    \node at (1.10,3.55) {event $n$};
    \node at (7.00,3.55) {event $n+1$};

    \node (j1n)  at (1.10,2.70) {$(X_n,T_n^{\vartheta_1})$};
    \node (j1np) at (7.00,2.70) {$(X_{n+1},T_{n+1}^{\vartheta_1})$};
    \node (xn)   at (1.10,1.25) {$X_n$};
    \node (xnp)  at (7.00,1.25) {$X_{n+1}$};
    \node (j2n)  at (1.10,-0.25) {$(X_n,T_n^{\vartheta_2})$};
    \node (j2np) at (7.00,-0.25) {$(X_{n+1},T_{n+1}^{\vartheta_2})$};

    \draw[->] (j1n) -- node[lab,above=2pt]
      {induced by $\mathcal J^{\vartheta_1}_{s_n,K_n}$} (j1np);
    \draw[->] (xn) -- node[lab,below=2pt]
      {selected event law $K_n$} (xnp);
    \draw[->] (j2n) -- node[lab,below=2pt]
      {induced by $\mathcal J^{\vartheta_2}_{s_n,K_n}$} (j2np);

    \draw[->] (j1n) -- node[lab,right=2pt]
      {state marginal ($\vartheta_1$)} (xn);
    \draw[->] (j1np) -- node[lab,left=2pt]
      {state marginal ($\vartheta_1$)} (xnp);
    \draw[->] (j2n) -- node[lab,right=2pt]
      {state marginal ($\vartheta_2$)} (xn);
    \draw[->] (j2np) -- node[lab,left=2pt]
      {state marginal ($\vartheta_2$)} (xnp);
\end{tikzpicture}
\caption{Reflexive calendar embedding. (a) The aggregate state $M_n$ determines the actor state $s_n=\mathcal I(M_n)$. The top-down constraint restricts the selected event kernel to $K_n\in\mathcal D(s_n)$; together $s_n$ and $K_n$ determine the compatible timing family $\mathfrak J(s_n,K_n)$ in Eq.~\eqref{eq:compatible-joint-family}. (b) Two choices $\vartheta_1,\vartheta_2\in\Theta$ give distinct joint state/waiting-time kernels $\mathcal J^{\vartheta_j}_{s_n,K_n}\in\mathfrak J(s_n,K_n)$. Both have event-state marginal $K_n$ by Eq.~\eqref{eq:compatible-state-marginal}, but may generate different holding times and calendar event epochs, giving different calendar embeddings through Eq.~\eqref{eq:compatible-calendar-embedding}; see Appendix~\ref{app:compatible-timing}. In the deterministic operational-clock specialisation of \citet{AngstmannGebbieAdjoint2026}, calendar time $t$ maps to operational time $u=U(t)$ with activity rate $\alpha(t)=\mathrm dU(t)/\mathrm dt$, while HC event coordination is $n_\ell=U_\ell(m)$.}
\label{fig:reflexive-calendar-embedding}
\end{figure}
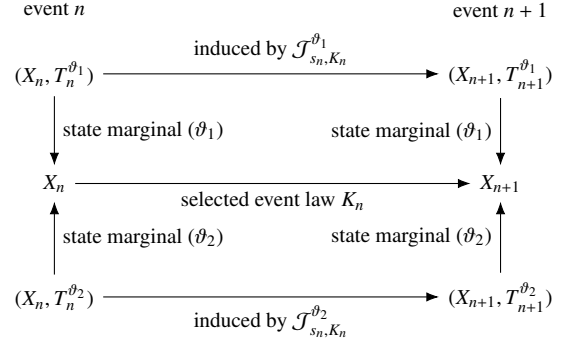

A calendar coordinate remains necessary for valuation, risk and comparison, but imposing synchronous global updating at the outset can obscure how hierarchical constraint changes event ordering, waiting times and the observed feedback.

\section{Markovianity and memory}
\label{sec:markov-memory}

Suppose the augmented event state $Z_n$ takes values in a measurable state space $(\mathsf Z,\mathcal Z)$ and contains the lower-level state together with any selection and clock variables required to determine the joint event kernel. The actor state is recovered from $s_n=\mathcal I(\Pi(X_n))$. The event-indexed system is Markov when, for every $B\in\mathcal Z$,
\begin{equation}
\Pr(Z_{n+1}\in B\mid Z_0,\ldots,Z_n)
=
\Pr(Z_{n+1}\in B\mid Z_n).
\label{eq:augmented-markov}
\end{equation}
The finite Markov kernels used in \citet{Gebbie2026} give a simple example of this reduction, with higher-level actor instances restricting the transitions that remain admissible.

Markovianity of the observed calendar-time process depends on the clock and on which variables are retained. A deterministic clock can preserve a Markov representation, and independent direct subordination of a Markov semigroup by a subordinator yields another Markov semigroup. Renewal or inverse-time constructions generally retain additional age or waiting-time information and can produce non-Markovian calendar-time equations; endogenous dependence of the clock on actor or constraint states adds another source of history dependence \citep{MeerschaertStraka2013,AngstmannGebbie2026}. Under standard semi-Markov conditions a Markov representation may be recovered by enlarging the calendar state to include an age or residual-time variable. The relevant distinction is between Markovianity of a sufficient enlarged state and Markovianity of the observed calendar projection. Reflexivity is therefore compatible with Markovian local event dynamics.

\section{Discussion}
\label{sec:discussion}

Lower-level event dynamics generate an aggregate state that an actor or institutional role interprets to restrict the transition kernels that remain admissible. The same actor state may also alter the joint waiting-time law by which those events enter calendar time. This gives a hierarchical-causal reading of the malleable causal chains described by \citet{PolakowGebbieFlint2026}: the admissible law itself depends on actor state. An admissible family is not yet a realised event law, and a realised event law does not by itself determine its calendar embedding. If either correspondence remains set-valued, the event-level specification need not determine a canonical calendar-time causal law, even when a sufficient event state is Markovian.

The main claim is that reflexivity is realised through endogenous top-down constraint: the higher-level state produced by lower-level dynamics changes the rules under which subsequent lower-level dynamics are generated. Non-unique time is a separate representation problem because compatible event dynamics need not determine a canonical global calendar description. Projection, clock state, age dependence or non-unique timing can then introduce memory or causal ambiguity in the observed calendar process without requiring non-Markov primitive event dynamics.

Two implications follow. First, invariant causal descriptions become conditional on the actor-level constraint. Standard econometric causal models ordinarily specify a structural or predictive relation and ask how variables behave under intervention or through time. Here $M_n$ enters $s_n$, which changes the correspondence $\mathcal D(s_n)$ of admissible lower-level transition laws even when an enlarged event state admits a time-homogeneous Markov representation. Interventions on actor roles, institutional constraints or market organisation may therefore change the transition correspondence rather than only the observed market variables. Calendar-time causal discovery that omits this state dependence can admit competing descriptions or misattribute a change in the transition correspondence to an observed-state effect. Hierarchical reflexivity differs accordingly from feedback introduced only through an additional drift, force or response term: the higher level changes the admissible lower-level dynamics.

Second, separating event dynamics from their calendar embedding creates a representation problem for quantitative finance. A selected event kernel $K$ need not determine a unique joint law for the next state and its waiting time, and hence need not determine a canonical calendar-time price process. This precedes the ordinary question of market completeness or uniqueness of an equivalent pricing measure because it concerns which calendar representation is being priced. Once a traded-market structure is specified, unresolved clock or projection risk may contribute to incompleteness; \citet{AngstmannGebbieAdjoint2026} use adjoint consistency between forward density evolution and backward valuation operators to restrict the admissible projection class. The same separation allows memory to appear in calendar observations without intrinsic long-range dependence in the event-indexed dynamics. Non-exponential or state-dependent holding times lead naturally to a semi-Markov description, while projection from $(\widehat Z_t,e_t)$ to $Y_t$ can discard variables on which the future law depends. Persistent or heavy-tailed timing can therefore generate calendar-time memory or heavy-tailed increments; volatility clustering additionally requires that persistence in the clock, state dependence or projection mechanism be transmitted into observed returns or volatility.

\section*{Acknowledgements}

I thank Daniel Polakow, George Ellis, Chris Angstmann, Diane Wilcox, and Andrew Paskaramoorthy for various conversations about causation, hierarchies and feedbacks in financial and complex systems.

\section*{AI Disclaimer}

OpenAI ChatGPT (GPT-5.6 Sol, accessed August 2026) was used for language editing, consistency checks, \LaTeX{} preparation and diagram refinement. I critically reviewed and edited all outputs and remain responsible for the mathematical formulation, interpretation and final text. 

\section*{Data availability}

No new data were created or analysed in this study.

\section*{Conflict of Interest}

I declare no conflict of interest.

\bibliography{references-v4.4.0}

\clearpage
\appendix
\section{Set-valued constraints}
\label{app:set-valued}

Let $(\mathsf X,\mathcal X)$ be the lower-level measurable state space, let $X_n\in\mathsf X$ denote the lower-level state at event $n$, let $(\mathsf M,\mathcal M)$ be the aggregate-state space, let $(S_A,\mathcal S_A)$ be the actor-state space, and let $\mathcal K(\mathsf X)$ be the set of Markov transition kernels on $\mathsf X$. In Hierarchical Causality \citep{Gebbie2026}, the fixed actor instance $A$ induces a correspondence $D^A_{r\to\ell}$ from higher-level states to admissible target-subsystem kernels. We assume here that this dependence acts through the actor-state sequence $(s_n)_{n\ge0}$ with $s_n=\mathcal I(M_n)$ and write $\mathcal D(s_n)$ for the resulting admissible family. Kernels are written on the full state space $\mathsf X$ for economy; the same correspondence may constrain the relevant local, restricted, marginal or conditional kernel section on the target subsystem.

Let $\Pi:(\mathsf X,\mathcal X)\to(\mathsf M,\mathcal M)$ be the measurable aggregation map and $\mathcal I:(\mathsf M,\mathcal M)\to(S_A,\mathcal S_A)$ the measurable interpretation map. For a lower-level state $x\in\mathsf X$, write $s\in S_A$ for the corresponding generic actor-state value,
\begin{equation}
s=\mathcal I\!\left(\Pi(x)\right).
\label{eq:appendix-actor-map}
\end{equation}
The higher-level constraint is represented by a non-empty correspondence
\begin{equation}
\mathcal D:S_A\rightrightarrows\mathcal K(\mathsf X),
\qquad
\mathcal D(s)\neq\varnothing,
\label{eq:appendix-correspondence}
\end{equation}
where $\mathcal D(s)$ is the set of lower-level transition kernels that remain admissible when the actor state has value $s$. The correspondence describes which event laws remain admissible; it is not itself a realised transition law. A realised stochastic model requires one admissible kernel to be selected or otherwise determined.
Set-valued transition laws also arise in imprecise Markov chains, where uncertainty is represented by credal sets of transition probabilities \citep{deCoomanEtAl2009}. Action-dependent families of transition laws are likewise standard in Markov decision processes \citep{Puterman1994}. Here the set is conditioned by actor state and is given a hierarchical causal interpretation; no reward or optimisation structure is assumed.

\begin{definition}[Selection]
A state-based selection is one possible resolution of the admissible family. It is a map
\begin{equation}
\sigma:S_A\longrightarrow\mathcal K(\mathsf X),
\label{eq:appendix-selection}
\end{equation}
where $\sigma(s)\in\mathcal D(s)$ for each $s\in S_A$. For the induced stochastic process, we require $(s,x)\mapsto \sigma(s)(x,B)$ to be $\mathcal S_A\otimes\mathcal X$-measurable for every $B\in\mathcal X$. No general measurable-selection theorem is assumed here. When a realised process is considered, a selection satisfying this measurability condition is treated as part of the model specification.
\end{definition}

Once such a selection has been fixed, the admissible family has been resolved into a single-valued next-event kernel. For $x\in\mathsf X$ and $B\in\mathcal X$,
\begin{equation}
\Pr(X_{n+1}\in B\mid X_n=x)
=
\sigma\!\left(\mathcal I(\Pi(x))\right)(x,B).
\label{eq:appendix-selected-law}
\end{equation}
The selected kernel gives a single-valued next-event transition law.

\begin{proposition}
Suppose $\mathcal D(s)$ is non-empty for every $s\in S_A$.
\begin{enumerate}[label=(\roman*)]
\item If $\mathcal D(s)$ is a singleton for every $s$, then the actor state determines a unique lower-level transition kernel.
\item If, for some $x\in\mathsf X$ and its actor state $s=\mathcal I(\Pi(x))$, there are $K_1,K_2\in\mathcal D(s)$ with $K_1(x,B)\neq K_2(x,B)$ for some $B\in\mathcal X$, then the actor state alone does not determine a unique next-event law. A selection or further state variable is required.
\end{enumerate}
\label{prop:selection}
\end{proposition}

\begin{proof}
The first statement follows directly from singleton-valuedness. For the second, the state $x$ generates the actor state $s=\mathcal I(\Pi(x))$, and that same actor state admits two kernels assigning different probabilities to the same measurable next-state event $B$ from $x$. The next-event law is therefore not fixed by $s$ alone.
\end{proof}

The proposition concerns only the distinction between admissibility and a realised next-event law. Existence of a measurable selection and empirical distinguishability of different admissible kernels require additional assumptions. No linear structure is assumed for $\mathcal D$.

\section{Event and calendar time}
\label{app:event-calendar}

Let $(\mathsf Z,\mathcal Z)$ be the sufficient event-state space and let $Z_n\in\mathsf Z$ contain the event-level variables needed to determine both the next state and its waiting time, including any selection or clock state when these variables affect the event law. The actor state is recovered from the lower-level state through $s_n=\mathcal I(\Pi(X_n))$.

The purpose of the augmented event state is to distinguish Markovianity in event index from Markovianity after calendar embedding. The former does not by itself imply the latter. Let $\tau_{n+1}>0$ denote the holding time from event $n$ to event $n+1$, let $\mathcal Q$ denote the joint kernel for the next sufficient event state and holding time, and let $\mathcal F_n$ be the $\sigma$-field generated by $(Z_0,\tau_1,Z_1,\ldots,\tau_n,Z_n)$. For $E\in\mathcal Z$ and Borel $I\subset(0,\infty)$,
\begin{equation}
\Pr\!\left(Z_{n+1}\in E,\,\tau_{n+1}\in I\mid\mathcal F_n\right)
=\mathcal Q(Z_n,E\times I).
\label{eq:appendix-joint-kernel}
\end{equation}
This is a Markov-renewal specification in the sense of \citet{Pyke1961}. The kernel $\mathcal Q$ is the sufficient-state representation of the body kernel $\mathcal J$ once the variables needed to determine the selected event law and timing specification are included in $Z_n$. For calendar time $t\ge0$, define the event epochs $T_n$ and the counting process $N(t)$ by
\begin{equation}
T_0=0,\qquad
T_n=\sum_{j=1}^{n}\tau_j,\qquad
N(t)=\max\{n:T_n\le t\}.
\label{eq:appendix-event-epochs}
\end{equation}
We assume $T_n\to\infty$ almost surely. Define the calendar process on the sufficient event state and its age by
\begin{equation}
\widehat Z_t=Z_{N(t)},
\qquad
e_t=t-T_{N(t)},
\label{eq:appendix-age-process}
\end{equation}
where $e_t$ is the elapsed time since the most recent event. Let $\pi_X:\mathsf Z\to\mathsf X$ denote the lower-level projection. The lower-level observable is $Y_t=\pi_X(\widehat Z_t)$.
There are therefore three potentially different objects: the event-indexed state $Z_n$, the calendar-time state augmented by its age, and the lower-level observable $Y_t$. Markovianity of one does not automatically imply Markovianity of the next.

\begin{proposition}
Suppose $\mathcal Q$ is time homogeneous, the waiting times are strictly positive, and the event process is non-explosive.
\begin{enumerate}[label=(\roman*)]
\item The calendar process $\widehat Z_t$ is semi-Markov. Under standard measurability and regularity conditions for the semi-Markov kernel, $(\widehat Z_t,e_t)$ has a Markov representation.
\item If, conditional on the current event state, the residual law of the next jump time and destination is independent of the current age, then the age variable is unnecessary and $\widehat Z_t$ is Markov in that state. For general holding-time laws it need not be.
\item A further projection $Y_t=\pi_X(\widehat Z_t)$ is not automatically Markov. It is Markov when the projected state is sufficient for the future law; clock state or age can otherwise remain dynamically relevant.
\end{enumerate}
\label{prop:calendar-markov}
\end{proposition}

\begin{proof}
Equation~\eqref{eq:appendix-joint-kernel} makes the next event state and holding time conditionally depend on the event history only through the current sufficient event state. Conditioning additionally on $e_t$ determines the residual joint law of the next jump time and destination, giving the Markov representation in (i). Age-independence of that residual law gives (ii). The last statement follows because projection can discard variables on which the conditional future law still depends.
\end{proof}

Direct subordination is more restrictive. Let $\widetilde X_u$ be a time-homogeneous Markov process with semigroup $(P_u)_{u\ge0}$. Let $S_t$ be an independent subordinator with law $\mu_t$ and Laplace exponent $\phi$. For Laplace variable $\lambda\ge0$, write $\mathbb E[e^{-\lambda S_t}]=e^{-t\phi(\lambda)}$. For bounded measurable $f$ and initial state $x$, with $\mathbb E_x$ denoting expectation conditional on $\widetilde X_0=x$,
\begin{equation}
P_t^{\phi}f(x)
=\mathbb E_x\!\left[f(\widetilde X_{S_t})\right]
=\int_{0}^{\infty}P_u f(x)\,\mu_t(\mathrm du).
\label{eq:appendix-direct-subordination}
\end{equation}
Here $(P_t^{\phi})_{t\ge0}$ is the directly subordinated semigroup in the Bochner sense \citep{Schilling1998}. An inverse subordinator or a renewal clock does not in general have the same semigroup composition and can generate calendar-time memory \citep{MeerschaertStraka2013}. A time-inhomogeneous event model similarly requires an augmented-state or evolution-family formulation before the Markov-semigroup statement can be used.

Selection of an admissible event law from $\mathcal D(s_n)$ is distinct from selection of a compatible joint state/waiting-time law and its induced event-to-calendar embedding. These are distinct sources of ambiguity. Selection from $\mathcal D(s_n)$ concerns which event law is realised. Selection from, or identification within, a compatible family of joint state/waiting-time laws concerns how that event law enters calendar time. In neither case does event-time Markovianity alone settle the calendar-time representation.

\section{Compatible timing}
\label{app:compatible-timing}

Fix an event $n\ge0$ and its actor state $s_n\in S_A$, a selected event-state transition kernel $K\in\mathcal D(s_n)$, and a current lower-level state $x\in\mathsf X$.

Fixing the event-state kernel does not in general fix its calendar embedding. The same kernel $K$ may be the state marginal of several joint laws for the next state and its waiting time. Let $\Theta$ index timing specifications compatible with that selected event-state law. Define the compatible family of joint state/waiting-time kernels by
\begin{equation}
\mathfrak J(s_n,K)
=
\left\{\mathcal J^{\vartheta}_{s_n,K}:\vartheta\in\Theta\right\}.
\label{eq:compatible-joint-family}
\end{equation}
Each $\mathcal J^{\vartheta}_{s_n,K}$ is a kernel from $\mathsf X$ to $\mathsf X\times(0,\infty)$ and must preserve the selected event-state marginal. Thus, for every $B\in\mathcal X$,
\begin{equation}
\mathcal J^{\vartheta}_{s_n,K}
\!\left(x,B\times(0,\infty)\right)
=
K(x,B),
\qquad \vartheta\in\Theta.
\label{eq:compatible-state-marginal}
\end{equation}
The members of $\mathfrak J(s_n,K)$ therefore share the same selected event-state marginal while differing in their holding-time marginal, in the dependence between the next state and its holding time, or in both. A fixed event kernel need not, for this reason, determine a unique calendar-time law.

For the fixed actor state $s_n$, collect both stages of local specification into
\begin{equation}
\mathfrak C(s_n)
=
\bigcup_{K\in\mathcal D(s_n)}\mathfrak J(s_n,K),
\label{eq:composite-event-timing-family}
\end{equation}
where each member of $\mathfrak C(s_n)$ is a joint state/waiting-time kernel whose state marginal is one admissible event kernel.

\begin{proposition}[Local event/timing uniqueness]
Fix $n\ge0$ and $s_n\in S_A$. Suppose $\mathcal D(s_n)$ is non-empty and $\mathfrak J(s_n,K)$ is non-empty for every $K\in\mathcal D(s_n)$. Then $\mathfrak C(s_n)$ is a singleton if and only if there is a unique admissible event kernel $K^\star$ and a unique compatible joint state/waiting-time kernel for it,
\begin{equation}
\mathcal D(s_n)=\{K^\star\},
\qquad
\mathfrak J(s_n,K^\star)=\{\mathcal J^\star\}.
\label{eq:local-event-timing-uniqueness}
\end{equation}
Consequently, non-uniqueness of either the admissible event kernel or the compatible timing law is sufficient to leave the local joint event/timing specification non-unique.
\label{prop:local-event-timing-uniqueness}
\end{proposition}

\begin{proof}
If both sets in Eq.~\eqref{eq:local-event-timing-uniqueness} are singletons, Eq.~\eqref{eq:composite-event-timing-family} gives $\mathfrak C(s_n)=\{\mathcal J^\star\}$. Conversely, suppose $\mathfrak C(s_n)$ is a singleton. Every $\mathcal J\in\mathfrak J(s_n,K)$ has state marginal $K$ by Eq.~\eqref{eq:compatible-state-marginal}. Distinct admissible kernels therefore cannot contribute the same joint kernel, since a joint kernel has a unique state marginal. Hence $\mathcal D(s_n)$ is a singleton. With its unique member $K^\star$ fixed, $\mathfrak J(s_n,K^\star)$ must also be a singleton.
\end{proof}

For a timing rule $\vartheta$ applied along the event sequence, let $(X_n^{\vartheta},\tau_n^{\vartheta})$ denote a realisation generated by the corresponding sequence of compatible joint kernels. Define
\begin{equation}
\begin{aligned}
T_0^{\vartheta}&=0,
& T_n^{\vartheta}&=\sum_{j=1}^{n}\tau_j^{\vartheta},\\
N_{\vartheta}(t)&=\max\{n:T_n^{\vartheta}\le t\},
& Y_t^{\vartheta}&=X_{N_{\vartheta}(t)}^{\vartheta}.
\end{aligned}
\label{eq:compatible-calendar-embedding}
\end{equation}
The superscript records that equality of the event-state marginal does not require the different timing specifications to use the same realised event chain. No common coupling across $\vartheta$ is assumed.

The construction is local. It identifies the joint timing laws compatible with a fixed actor state $s_n$ and selected event kernel. It does not prescribe how admissible kernels or compatible timing laws are resolved as the actor state evolves, nor does it settle how the resulting local descriptions should be related to a global calendar representation. Although the event-state marginal is held fixed, the superscript is retained because the corresponding state chains need not be identical realisations under different joint state/waiting-time laws. The remaining issue is therefore a local-to-global one: how the compatible timing descriptions are related to a global calendar representation is left open.\footnote{Some care is needed here: selection of a representative, projection to a reduced description, and quotienting by an equivalence relation are mathematically different operations. The present paper leaves open which, if any, is appropriate for the global calendar representation. In the pricing construction of \citet{AngstmannGebbieAdjoint2026}, the problem is specialised to operational-to-calendar projections, with adjoint consistency subsequently restricting the admissible projection class.}

This is the pre-pricing calendar-representation problem. In \citet{AngstmannGebbieAdjoint2026}, the operational kernel is held fixed while a family of operational-to-calendar projections is varied; for pricing, candidate projections are then restricted by the requirement that the induced forward density and backward valuation operators remain adjoint on the chosen state space. Here the analysis stops one layer earlier: it identifies the local family of event/timing laws compatible with a selected event kernel, without requiring that this family already determine a unique global calendar representation, pricing measure, or adjoint-real admissibility condition.
\end{document}